\documentclass[12pt]{iopart}
\expandafter\let\csname equation*\endcsname\relax
\expandafter\let\csname endequation*\endcsname\relax
\usepackage{amssymb}
\usepackage{amsmath}
\usepackage{amsthm}
\usepackage{thmtools}   
\usepackage[numbers,sort&compress,square]{natbib}
\usepackage{bm}
\usepackage{graphicx}
\usepackage{calc}
\usepackage[dvipsnames]{xcolor}
\usepackage{hyperref}
\usepackage{tabularray}
\usepackage[capitalise]{cleveref}

\newtheorem{theorem}{Theorem}[section]
\newtheorem{lemma}[theorem]{Lemma}
\newtheorem{proposition}[theorem]{Proposition}

\newtheorem{conjecture}{Conjecture}

\theoremstyle{definition}
\newtheorem{definition}[theorem]{Definition}

\newcommand{\Vin}{V_\text{in}}

\newcommand{\Ein}{E_\text{in}}

\newcommand{\Bogo}{Bogoyavlenskij~}
\newcommand{\Evri}{Evripidou~}

\begin{document}

\title[New Families of Integrable Graphs]{New Families and Operations on Integrable Graphs}

\author{Matthew Visomirski$^1$ and Christopher Griffin$^{2,*}$}
\address{
	$^1$Department of Physics,
	University of Texas at Austin,
    Austin, TX 78705
    }
\address{
	$^2$Applied Research Laboratory,
	The Pennsylvania State University,
    University Park, PA 16802
    }
\address{$^*$Corresponding Author}

\eads{\mailto{mv33573@my.utexas.edu}, \mailto{griffinch@psu.edu}}
\date{\today~-~Preprint}

\begin{abstract} In this paper, we investigate the integrability of Lotka-Volterra (replicator) systems arising from interaction matrices generated from corresponding graph structures, continuing work started by Visomirski and Griffin [J. Phys. A., 58:015701, 2025] and \Evri et al. [J. Phys. A., 55:325201, 2022] (among others). In particular, we define a new family of graphs, the spoked graphs, and show that all dynamics generated from this family are integrable. In reference to \Evri et al. (2022), we define a new anti-cloning operator and show that its action on balanced tournament graphs (with odd vertex count) generates new graphs whose dynamics are integrable. Interestingly, we provide numerical evidence that this anti-cloning operation leads to chaotic behaviour when applied to other graph families (e.g., the directed cycles that generate the classically integrable Volterra lattice). This work completes a taxonomy of all integrable dynamics generated by directed graphs with up to six vertices started by Visomirski and Griffin (2025), and suggests several future directions of study on this topic.
\end{abstract}

\vspace{2pc}

\submitto{\jpa}
\maketitle

\section{Introduction}
In 1976, Smale \cite{S76} noted the potential complexity of the dynamics of the generalised Lotka-Volterra dynamics, stating that these simple-looking systems can produce any conceivable trajectory behaviour. Surprisingly, almost a century ago Volterra's work \cite{volterra1931leccons,volterra1937principes,volterra1931variations} as well as the myriad of work that followed \cite{christodoulidi2019new,fernandes1995hamiltonian,visomirski2025integrability,bountis2021integrable,peixe2015lotka,van2025quadratic,charalambides2015generalized,I87,I08,BIY08,EKV21,EKV22,PG23} shows that while these systems can indeed produce highly complex dynamics, they have hidden within them infinite families of integrable dynamics. This fact is even more surprising given the general rarity of integrable systems, in the sense that a randomly chosen system of differential equations is unlikely to be integrable, let alone exhibit multiple conserved quantities or even be Hamiltonian. 

In this paper, we focus on a subset of integrable Lotka-Volterra dynamics that are generated (in a formal sense) by directed graph structures. This allows us to define integrability preserving operations on the generating graph structure. In the strictest sense, the resulting class of Lotka-Volterra dynamics are replicator equations (usually found in evolutionary game theory \cite{T15,HS98,HS03}). Let $\mathbf{A} \in \mathbb{R}^{n \times n}$ be an interaction matrix. The replicator dynamics \cite{HS98,HS03} are the system of differential equations with form,
\begin{equation*}
    \dot{u}_i = u_i\left(\mathbf{e}_i^T\mathbf{A}\mathbf{u} - \mathbf{u}^T\mathbf{A}\mathbf{u}\right),
\end{equation*}
where $\mathbf{u} = \langle{u_1,\dots,u_n}\rangle$ is usually thought of as a vector of species proportions and $\mathbf{e}_i$ is the $i^\text{th}$ standard basis vector. As such $\mathbf{x} \in \Delta_{n-1}$, where,
\begin{equation*}
    \Delta_{n-1} = \left\{\mathbf{u} \in \mathbb{R}^n : \sum_i u_i = 1, u_i \geq 0\right\},
\end{equation*}
is the $n-1$ dimensional unit simplex embedded in $\mathbb{R}^n$. Consequently, $H_0 = x_1 + \cdots + x_n$ is immediately seen to be a constant of motion for the replicator equations. In general, the replicator equations are known to be diffeomorphic to the generalised Lotka-Volterra equations using Hofbauer's trick (see \cite{HS98,HS03,H96}). When we restrict to the case that $\mathbf{A}$ is skew-symmetric, the replicator equations simplify to,
\begin{equation}
    \dot{u}_i = u_i\left(\mathbf{e}_i^T\mathbf{A}\mathbf{u}\right),
    \label{eqn:Replicator}
\end{equation}
because $\mathbf{u}^T\mathbf{A}\mathbf{u} = 0$ and the resulting replicator equations are an instance of Lotka-Volterra equations. If we further restrict to the case that $\mathbf{A}$ is skew-symmetric having only entries drawn from the set $\{-1,0,1\}$, then the resulting systems of differential equations can be put into one-to-one correspondence with directed graphs. Let $\mathbf{A} \in \{-1,0,1\}^{n \times n}$ be such a skew-symmetric matrix and let $G_\mathbf{A} = (V,E)$ be a directed graph defined so that $V = \{1,\dots,n\}$ and $E \subset V \times V$ with $(i,j) \in E$ if and only if $A_{ij} = -1$. Biologically speaking, the edge $(i,j)$ indicates that species $i$ is consumed by species $j$. This correspondence establishes the relationship between the directed graphs and the instances of the Lotka-Volterra (replicator) dynamics discussed in this paper. As such, we will often refer to a graph as having a property (e.g., integrability) to mean that its corresponding dynamical system has this property. Determining the interplay between combinatorial properties of the (corresponding) graph structures and the properties of the dynamical systems is a differentiator within the literature on integrability of the Lotka-Volterra system, with some work focusing on the class of systems having graph-theoretic representations (see e.g., \cite{I87,I08,BIY08,PG23,visomirski2025integrability,EKV21,EKV22})and other authors focusing on a distinct class of systems (see e.g., \cite{bountis2021integrable,christodoulidi2019new,fernandes1995hamiltonian,peixe2015lotka,van2025quadratic,charalambides2015generalized}).

In this paper, we restrict our attention to those Lotka-Volterra (replicator) dynamics that can be put into one-to-one correspondence with directed graphs and expand on results first discussed by Visomirski and Griffin in \cite{visomirski2025integrability}, while continuing to build on and use the prior work of Itoh \cite{I87,I08}, \Bogo, Itoh and Yukawa \cite{BIY08} and \Evri, Kassotakis and Vanhaecke \cite{EKV21,EKV22} and Griffin and Paik \cite{PG23}. The main results of this paper, which we make precise in the sequel, are as follows.
\begin{itemize}
    \item We define a \textit{spoking} operation on a class of directed cycles to create a new (infinite) family of integrable graphs that we call the \textit{spoke graphs}. Representatives of this family were observed as outliers in the taxonomy of integrability provided in \cite{visomirski2025integrability}. Here, we formally show these examples are part of this larger integrable family.
    \item We define an \textit{anti-cloning} operation on the odd balanced tournament graphs and show that all members of this family are integrable. Anti-cloning is an homage to the work on cloning found in the work of Evripidou et al. \cite{EKV22}. 
    \item As a result, we can complete the taxonomy of integrable graphs with six or fewer vertices begun by Visomirski and Griffin \cite{visomirski2025integrability}.
\end{itemize}
As before, our ultimate goal is to characterize the behaviour of the replicator dynamics arising from skew-symmetric matrices (with entries in $\{-1,0,1\}$) as a consequence of the combinatorics of their corresponding graphs. This paper takes an important step, effectively identifying all integrable families represented in graphs with 6 or fewer vertices and suggesting that novel results will require either a more (exhaustive) experimental analysis of (e.g.) graphs with 7 vertices or a completely new theoretical approach.

The remainder of this paper is organized as follows: In \cref{sec:Prelim} we provide preliminary definitions and relevant prior results. Results on spoke graphs are provided in \cref{sec:Spoke}. We analyse the anti-cloning operation in \cref{sec:AntiCloning}. A conjecture on the class of \Bogo graphs is provided in \cref{sec:Bogo}. In \cref{sec:Taxonomy} we present the final taxonomy of behaviour for all graphs with six or fewer vertices. Conclusions and future directions are presented in \cref{sec:Conclusions}. 

\section{Preliminary Definitions and Prior Results}\label{sec:Prelim}
We provide mathematical preliminaries, many of which can be found in the work of Visomirski and Griffin \cite{visomirski2025integrability} and Evripidou et al. \cite{EKV22}. Additional information on Poisson structures is provided in Laurent-Gengoux, Pichereau and Vanhaecke's book on the subject \cite{L-GPV12}.

\subsection{Graph Operations}
We formalize the relationship between directed graphs and skew-symmetric matrices with entries drawn from the set $\{-1,0,1\}$. First, note that in the directed graph $G = (V,E)$ the oriented edge $(i,j)$ either does not appear in $E$ or if it does appear, then the edge $(j,i)$ does not appear in $E$. Throughout this work, we assume all graphs are simple in the sense that they have no multi-edges or self-loops (see Griffin \cite{griffin2023applied} for general details on graphs). We now adopt the following results from \Evri et al. \cite{EKV22} and our prior work \cite{visomirski2025integrability}.

\begin{definition}[Skew Symmetric Graph] Let $\mathbf{A} \in \mathbb{R}^{n\times n}$ be a skew-symmetric matrix with entries taken from $\{-1,0,1\}$. The oriented directed graph $G_\mathbf{A} = (V,E)$ has vertex set $V = \{1,\dots,n\}$ and edge set $E \subset V \times V$ so that,
\begin{equation*}
    (i, j) \in E \iff A_{ij} = -1.
\end{equation*}
\end{definition}
We immediately have the  following result. 
\begin{proposition} There is a one-to-one correspondence between simple oriented directed graphs and skew-symmetric matrices with entries $0$,$\pm 1$. \hfill\qed
\label{prop:OneToOne}
\end{proposition}
\Evri et al. \cite{EKV22} refer to these graphs as  skew-symmetric graphs. As in \cite{visomirski2025integrability}, for the remainder of this paper, all graphs are skew-symmetric and are called graphs. We will use a special subset of these graphs, sometimes referred to as \Bogo-Itoh graphs, or sometimes just the \Bogo graphs \cite{EKV22}. 
\begin{definition}[\Bogo Graphs] The \textit{\Bogo graph} $B(n,k)$ where $k < \frac{n}{2}$ has vertex set $V = \{1,\dots,n\}$ and edge $E$ with edge $(i,j) \in E$ if $j = 1 + (i \oplus_n l)$ for $0 \leq l < k$ and $\oplus_n$ denotes addition modulo $n$ (corrected to start at 1).
\label{Bogo Def}
\end{definition}
If the numbers in $\{1,\dots,n\}$ are arranged in a circle, then the graph $B(n,k)$ has directed edges from vertex $i$ to the next $k$ vertices, working around the circle. This is illustrated in \cref{fig:B72}.
\begin{figure}[htbp]
\centering
\includegraphics[width=0.35\textwidth]{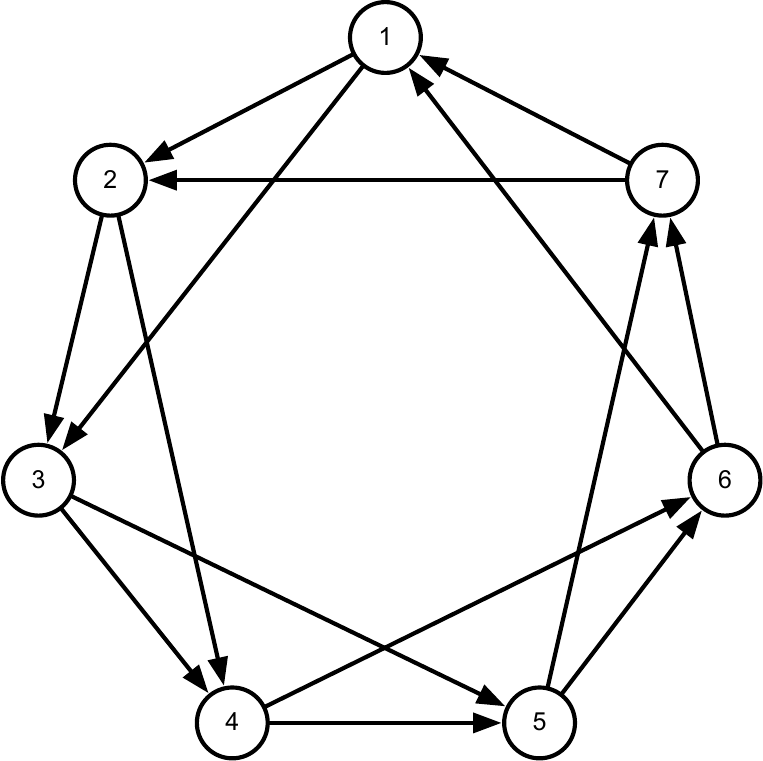}
\caption{The \Bogo graph $B(7,2)$ illustrates the definition of these graphs using geometric arrangement.}
\label{fig:B72}
\end{figure}

\begin{definition}[Balanced Tournament] A \Bogo graph $B(n,k)$ where $n = 2m+1$ (i.e., $n$ is odd) and $k = m$ is called a \textit{balanced tournament}. For simplicity, we denote the balanced tournament with $n$ vertices by $\mathrm{To}(n)$, where $n$ must be odd.
\end{definition}
It is straightforward to see that a graph is a balanced tournament if and only if there is a directed edge between any pair of vertices and for every vertex the number of outbound edges is equal to the number of inbound edges. Classic games like rock-paper-scissors are balanced tournaments. \cref{fig:BalancedTournament} (left) shows the classic rock-paper-scissors graph and is the only example of a directed cycle that is also a complete balanced tournament. \cref{fig:BalancedTournament} (right) is the graph of the game rock-paper-scissors-Spock-lizard \cite{kassRock2026}. Both are balanced tournaments.
\begin{figure}[htbp]
\centering
\includegraphics[width=0.45\textwidth]{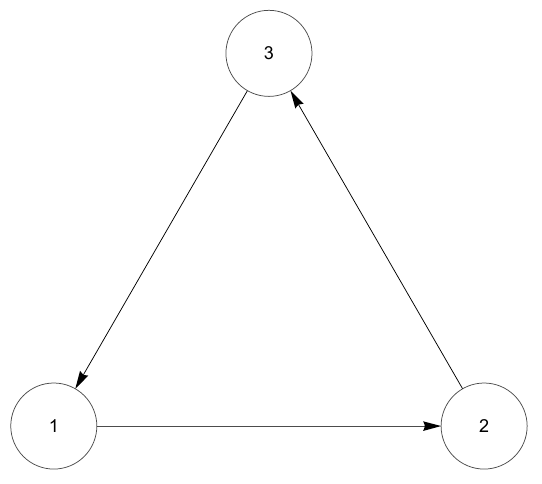}\quad
\includegraphics[width=0.45\textwidth]{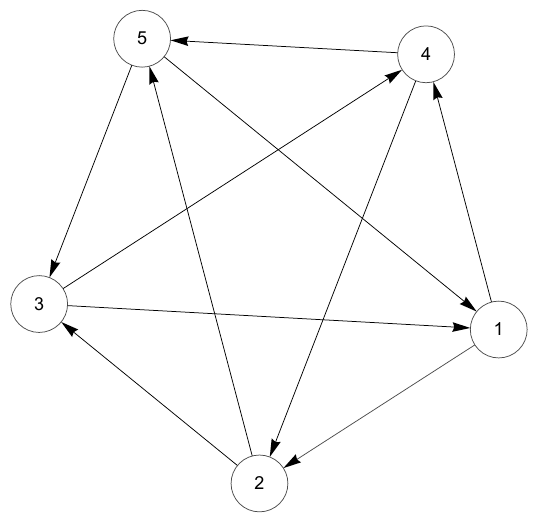}
\caption{(Left) The balanced tournament on three vertices. (Right) The balanced tournament on five vertices. }
\label{fig:BalancedTournament}
\end{figure}

We now formally define the \textit{spoking} operation on directed cycles, first observed anecdotally by Visomirski and Griffin \cite{visomirski2025integrability}.
\begin{definition}[Spoked Graphs] Let $G = (\Vin,\Ein)$ be a directed cycle on $n$ vertices with $n \geq 4$. Let $k$ be a positive integer with $k \leq \lfloor\tfrac{n}{3}\rfloor$ and let $i_1,\dots,i_k \in \{1,\dots,n\}$ be an ascending sequence with $1 \leq i_1 < i_2 < \dots i_k \leq n - 2$ and so that $i_{j+1} \geq i_j + 3$. Then the \textit{spoked cycle} on $n$ vertices with base vertices $i_1,\dots,i_k$, denoted, 
$\mathrm{Sp}(n,i_1,\dots,i_k)$ is (isomorphic to) the graph with vertex set $V = \{1,\dots,n, n+1,\dots,n+k\}$ and edge set $E$ with,
\begin{equation*}
    E = \Ein \cup \{(i_j, n+j), (i_j\oplus 1, n+j), (n+j, i_j\oplus 3), (n+j, i_j\oplus 4)\},
\end{equation*}
where $\oplus$ indicates addition modulo $n$, corrected for starting the count at $1$. This process is called \textit{spoking} a cycle and vertices $n+1,\dots,n+k$ are called \textit{spoke vertices}. We denote the resulting graph as $\mathrm{SG}(n,k)$
\label{def:SpokedCycle}
\end{definition}  
\cref{fig:Spoke} (left) shows an example of a cycle that is fully spoked. In this case, two spoke vertices have been added to the base cycle with six vertices, while \cref{fig:Spoke} (right) shows a single spoke vertex added to the seven cycle.
\begin{figure}[htbp]
\centering
\includegraphics[width=0.45\textwidth]{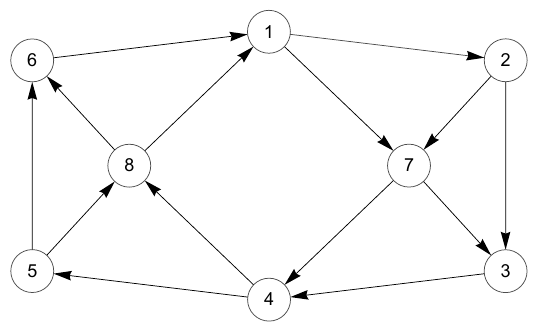} \quad \includegraphics[width=0.45\textwidth]{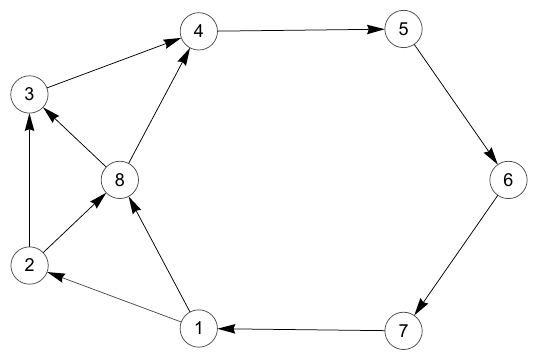}
\caption{(Left) A spoke graph with six vertices and the maximum number of spokes possible. (Right) A spoke graph with seven vertices and a single spoke vertex.}
\label{fig:Spoke}
\end{figure}
Notice from this definition if 

The second graph operation, anti-cloning, we study is a variation of the vertex cloning operation studied by \Evri et al. \cite{EKV22}. Two vertices of a graph are clones if they share identical neighbourhoods. That is, if the matrix $\mathbf{A}$ has at least two repeated rows (columns). In our case, we define a similar anti-cloning operation, but as we will see it will be relevant only when applied to balanced tournaments.

\begin{definition}[Anti-Cloning] Two vertices in a graph $G_\mathbf{A} = (V,E)$ are anti-clones if they have the same undirected neighbourhoods, but the direction of the in-bound and out-bound edges are opposite. That is, the matrix $\mathbf{A}$ has two rows (columns) corresponding to the two vertices that are additive inverses. \textit{Anti-cloning} a vertex $v \in V$ simply adds an anti-clone vertex to the graph. We denote this operation on $G_\mathbf{A}$ and vertices $v_1,\dots,v_k$ by $\mathrm{Ac}(G,\{v_1,\dots,v_k\})$. That is, in the resulting graph $v_1,\dots,v_k$ will have anti-cloned vertices.
\label{def:AntiCloning}
\end{definition}

\cref{fig:AntiCloning} shows an example of the anti-cloning of vertex 1 in a balanced tournament with five vertices. The anti-cloned vertex is vertex 6. This is one of three anomalous (integrable) graphs identified by Visomirski and Griffin \cite{visomirski2025integrability} in their taxonomy of Lotka-Volterra (replicator) systems defined by directed graphs (antisymmetric interaction matrices) with six species.
\begin{figure}
    \centering
    \includegraphics[width=0.5\linewidth]{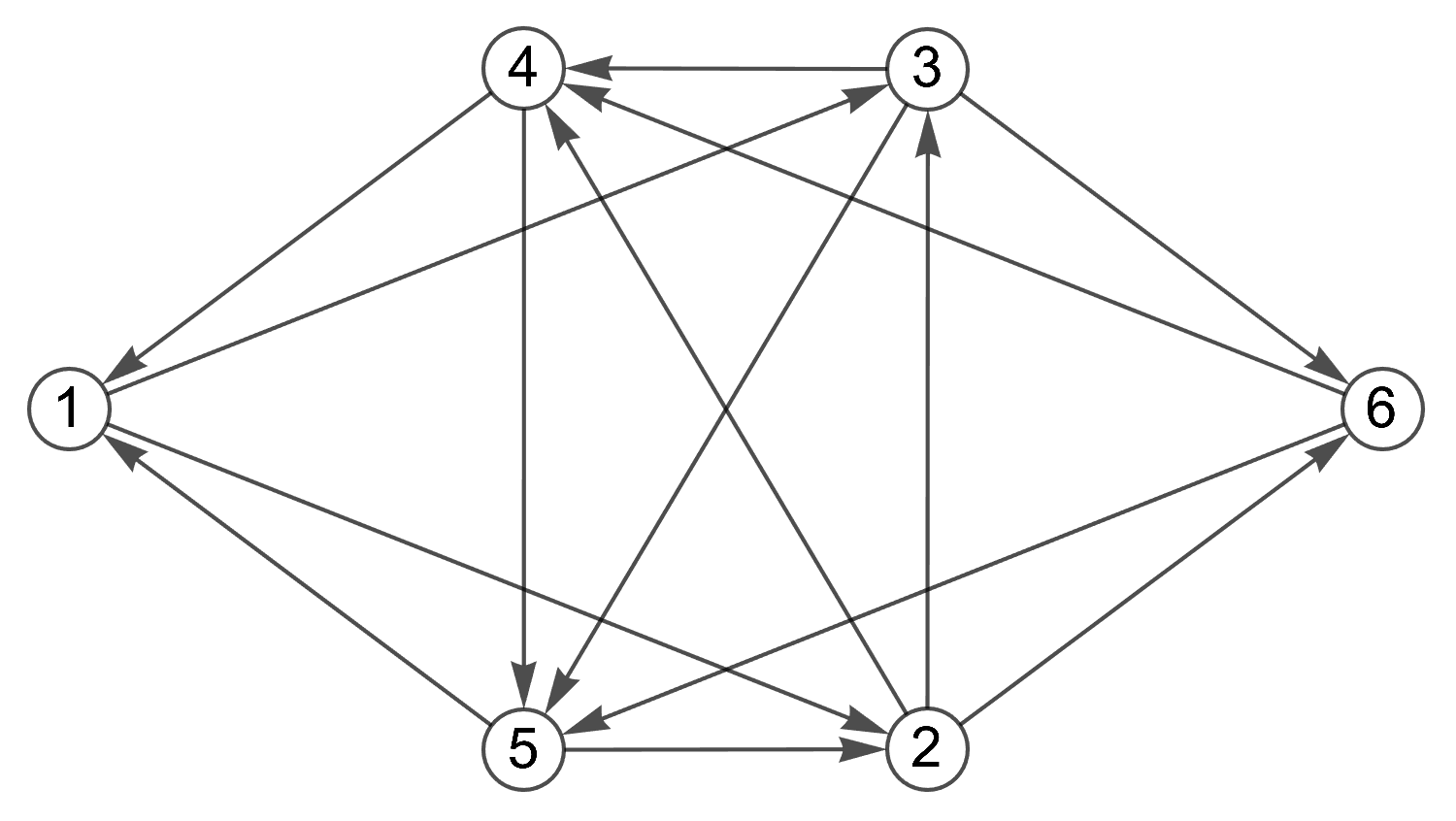}
    \caption{Seed Graph for Anti-cloning Family }
    \label{fig:AntiCloning}
\end{figure}

\subsection{Integrability of Replicator Systems}
Let $\mathbf{A} \in \mathbb{R}^{n \times n}$ be a skew-symmetric matrix and let $F,G:\mathbb{R}^n \to \mathbb{R}$ be differentiable functions. Using the quadratic bracket,
\begin{equation}
\{F,G\}_\mathbf{A} = \sum_{i < j} A_{ij}u_i u_j \left(\frac{\partial F}{\partial u_i} \frac{\partial G}{\partial u_j} - \frac{\partial G}{\partial u_i}\frac{\partial F}{\partial u_j}\right),
\label{eqn:Bracket}
\end{equation}
along with straightforward computation, \cite{G21}, shows that \cref{eqn:Replicator} is a Hamiltonian system with Hamiltonian $H_1(\mathbf{u}) = u_1 + \cdots + u_n$ and
\begin{equation*}
    \dot{u}_i = \{u_i, H_1\}_\mathbf{A} = u_i\left(\mathbf{e}_i^T\mathbf{A}\mathbf{u}\right).
\end{equation*}
Here, $H_1$ is a conserved quantity, i.e., $\dot{H}_1 = 0$, which is clear for the replicator as we automatically have $H_1 = 1$, assuming we begin with an initial condition on $\Delta_{n-1}$. Additional conserved quantities admitted by specific cases of the Lotka-Volterra (replicator) dynamics will be discussed in the sequel and are already discussed extensively in the literature \cite{I87,I08, PG23,visomirski2025integrability}. It is worth noting that the quadratic bracket given in \cref{eqn:Bracket} has an interesting history of independent rediscovery, as discussed by Visomirski and Griffin \cite{visomirski2025integrability}.


The quadratic bracket admits Poisson structure,
\begin{equation*}
    \pi_\mathbf{A} = \sum_{i < j} A_{ij}u_i u_j\frac{\partial}{\partial u_i} \wedge \frac{\partial}{\partial u_j},
\end{equation*}
which has been studied by \Evri et al. \cite{EKV22} and is discussed in \cite{L-GPV12}. Interestingly, the rank of $\pi_\mathbf{A}$ corresponds to the rank of $\mathbf{A}$. This allows us to make the following definition.

\begin{definition}[Def. 12.9 of \cite{L-GPV12}] Let $\mathbf{A} \in \mathbb{R}^{n \times n}$ be a skew-symmetric matrix with $\mathrm{rank}(\mathbf{A})$. The replicator dynamics, \cref{eqn:Replicator}, generated from $\mathbf{A}$ are \textit{integrable} if there are $s$ conserved quantities $H_1,\dots,H_s$ so that:
\begin{enumerate}
    \item $H_1,\dots,H_s$ are algebraically independent.
    \item $H_1,\dots,H_s$ commute under the action of the bracket (or are in involution). That is, $\{H_i,H_j\}_\mathbf{A} = 0$ for all $i,j$. 
    \item The following relation between matrix rank (Poisson manifold rank), number of conserved quantities, and embedding dimension holds,
    \begin{equation*}
        s + \tfrac{1}{2}\mathrm{Rank}(\mathbf{A}) = n.
    \end{equation*}
\end{enumerate}
\label{def:LAIntegrability}
\end{definition}
In light of \cref{prop:OneToOne} we say that a skew-symmetric graph $G_\mathbf{A}$ is integrable if the resulting replicator dynamics generated from the (skew-symmetric) matrix $\mathbf{A}$ are integrable. In particular, Kac and Moser \cite{KM75} and Moerbeke \cite{M74} prove the following.



\begin{theorem} Every directed cycle is integrable. Consequently, each directed cycle admits a sufficient number of conserved quantities to be integrable.\hfill\qed
\end{theorem}
This result was originally phrased in terms of the dynamics alone, and the family of dynamics arising from the directed cycles is the Volterra lattice. 

More generally, the following theorem is immediate from the work of Itoh and  \Bogo \cite{I87,I08,BIY08}.
\begin{theorem} Let $G = B(n,k)$ be a \Bogo graph. Then $G$ admits a sufficient number of conserved quantities to be integrable. \hfill\qed
\label{theorem:bogoccs}
\end{theorem}
It is worth noting that the integrability of the \Bogo graphs is only known for directed cycles \cite{M74,KM75} and balanced tournaments \cite{VS93,I87,B91}. However, it is conjectured that this family of graphs is integrable and specific members have been shown to be integrable for the purposes of taxonomy (see e.g., \cite{visomirski2025integrability}).

Following Visomirski and Griffin \cite{visomirski2025integrability}, who in turn follow Itoh \cite{I87,I08} we will use combinatorial properties of the graphs to construct conserved quantities. Recall, the graph $B(n,1)$ is the directed cycle with $n$ vertices.  
In the dynamical system generated from a directed cycle, there are always at least two conserved quantities,
\begin{equation*}
H_1 = \sum_{i} u_i \qquad \text{and} \qquad H_{n} = \prod_{i} u_i.
\end{equation*}
Note the interaction matrix for the directed $n$ cycle has rank $n - 2$ when $n$ is even and rank $n - 1$ when $n$ is odd. Consequently, when $n$ is odd, there must be $\left\lceil\tfrac{n}{2}\right\rceil = \left\lfloor\tfrac{n}{2}\right\rfloor + 1$ conserved quantities and when $n$ is even, there must be $\tfrac{n}{2} + 1 = \left\lfloor\frac{n}{2}\right\rfloor + 1$ conserved quantities. Put more simply, if $n = 2m$ or $n = 2m+1$, there are $m+1$ conserved quantities.
 
Visomirski and Griffin \cite{visomirski2025integrability} define a \textit{non-edge} as a pair $(i,j)$ with $i < j$ so that neither $(i,j)$ nor $(j,i)$ is an element of $E$. Likewise, a \textit{non-cycle} is a sequence $(i_1,\dots,i_r)$ so that, $i_j < i_{j+1}$, and for all pairs $(i_j, i_k)$, the pair $(i_j, i_k)$ is a non-edge and the pair $(i_r,i_1)$ is also a non-edge. From this, it is clear that, $(i_1,\dots,i_r)$ forms a cycle in the graph complement of $B(n,1)$. If we think of non-edges as non-cycles of length $2$ and denote the set of all non-cycles of length $l$ as $\overline{\mathcal{C}}(l)$, then direct enumeration shows that the graph $B(n,1)$ possesses non-empty sets of non-cycles $\overline{\mathcal{C}}(k)$ for $2 \leq k \leq \left\lfloor\tfrac{n}{2}\right\rfloor$. Using these non-cycles, we define the quantities, 
\begin{equation}
    H_k = \sum_{c \in \overline{C}(k)} \prod_{j \in c} u_j.
    \label{eqn:ItohCC}
\end{equation}
The following is a restatement of the integrability of the Volterra lattice and follows from Itoh \cite{I08}.
\begin{theorem} For $k\in \left\{1,2,\dots,\left\lfloor{\tfrac{n}{2}}\right\rfloor,n\right\}$, the quantity $H_k$ is conserved in the replicator dynamics generated by $B(n,1)$. Moreover, the conserved quantities $H_k$ are (algebraically) independent and in involution with respect to the bracket given by \cref{eqn:Bracket}.
\hfill\qed
\end{theorem} 
Thus, in a directed $n$-cycle with $n=2m$ or $n = 2m+1$ we see there are $m-1$ conserved quantities generated by non-cycles and two additional conserved quantities, $H_1$ and $H_n$ making $m+1$ conserved quantities in all. 

Paik and Griffin \cite{PG23} and Itoh \cite{I87} note that the conserved quantities for the balanced tournament graph $\mathrm{To}(2m+1)$ are sums of the Casimirs of the Poisson structure corresponding to balanced sub-tournaments, or maximal \Bogo subgraphs. Formally, for $k \in \{1,\dots,m\}$, let $\mathcal{T}(2k+1)$ be the set of vertex subsets of form $\{i_1,\dots,i_{2k+1}\}$ in $\mathrm{To}(2m+1)$ so that the subgraph generated from $\{i_1,\dots,i_{2k+1}\}$ is a  balanced sub-tournament of $\mathrm{To}(2m+1)$. Then the conserved quantities of the dynamics generated by $\mathrm{To}(2m+1)$ are,
\begin{equation}
    H_{2k+1} = \sum_{t \in \mathcal{T}(2k+1)}\prod_{j\in t}u_j.
    \label{eqn:ItohCC2}
\end{equation}
The following is a restatement of the integrability of the dynamics generated by the balanced tournaments (see \cite{I87,PG23}).
\begin{theorem} For $k\in \left\{1,2,\dots,m+1\right\}$, the quantity $H_k$ is conserved in the replicator dynamics generated by $\mathrm{To}(2m+1)$. Moreover, the conserved quantities $H_k$ are (algebraically) independent, in involution with respect to the bracket given by \cref{eqn:Bracket} and composed of sums of Casimirs of the dynamics generated by the balanced sub-tournaments of $\mathrm{To}(2m+1)$.\hfill\qed
\label{thm:Tournament}
\end{theorem}

\Evri et al. \cite{EKV22} note that the Casimirs of the Poisson algebra corresponding to $\pi_\mathbf{A}$ can be read from the eigenvectors of $\mathbf{A}$ that correspond to the zero eigenvalue. To see this, note that if $\bm{\alpha} = \langle{\alpha_1,\dots,\alpha_n}\rangle$ is an eigenvector of $\mathbf{A}$ with eigenvalue $0$, then,
\begin{equation}
    \frac{d}{dt}\left(u_1^{\alpha_1} \cdots u_n^{\alpha_n}\right) = u_1^{\alpha_1} \cdots u_n^{\alpha_n}\left(\bm{\alpha}^T\mathbf{A}\mathbf{u}\right) = 0,
\label{eqn:TimeDerivativeOfProduct}
\end{equation}
and thus $x_1^{\alpha_1} \cdots x_n^{\alpha_n}$ is a conserved quantity of the system. In the case of $\Bogo$ graphs, these highest degree Casimirs have $\alpha_i = 1$ for $1\leq i \leq n$. 

For the graphs we consider, the resulting dynamics will have conserved quantities, all in a form that generalizes \cref{eqn:ItohCC,eqn:ItohCC2}. Consequently,  Visomirski and Griffin \cite{visomirski2025integrability} introduced the following definition in light of Itoh's contributions to the combinatorial connection between integrable Lotka-Volterra systems and graphs.
\begin{definition}[Itoh-Style Conserved Quantities] An oriented directed graph $G$ admits Itoh-style conserved quantities if it has a set of algebraically independent conserved quantities that can be written in the form,
\begin{equation}
    H = \sum_{j\in \mathcal{J}}\prod_{i\in\mathcal{I}(j)} u_i,
    \label{eqn:ItohStyle}
\end{equation}
for appropriate index sets $\mathcal{I}(j)$ and $\mathcal{J}$.
\end{definition}
In general, we expect the index set $\mathcal{J}$ to be defined combinatorially from the graph structure itself. By way of example, consider the tournament graph $\mathrm{To}(n)$ with $n = 2m+1$. Its conserved quantities can be derived from its sub-tournaments as shown by Vesselov-Shabbat \cite{VS93}, Itoh \cite{I87} and \Bogo \cite{B88,bogoyavlensky1988five}. In this case, let $\mathcal{T}(k)$ be the set of (vertex sets of) sub-tournaments of $\mathrm{To}(n)$ with $k \leq n$ vertices. Here, we assume sub-tournaments are modulo graph isomorphism. Then, 
\begin{equation*}
    H_k = \sum_{t \in \mathcal{T}(k)}\prod_{i\in t} u_i,
\end{equation*}
is a conserved quantity where $t$ is the set of vertices in the sub-tournament. 


\subsection{Graph Derivative Operator}
Griffin and Visomirski note that each graph defines a graph derivative on functions of trajectories $\mathbf{u}(t)$. Let $G_\mathbf{A}$ be a graph generated by (skew-symmetric) matrix $\mathbf{A}$. Then, 
\begin{equation*}
    D_{G_\mathbf{A}}[f] = \sum_i u_i\frac{\partial f}{\partial u_i} \mathbf{e}_i^T\mathbf{A}\mathbf{u}.
\end{equation*}
By way of example, for a directed cyclic graph $B(n,1)$ and a product $T = u_{i_1} \cdots u_{i_k}$, we have\footnote{Note, we have corrected a sign error in \cite{visomirski2025integrability}.},
\begin{equation}
    D_{B(n,1)}[T] = T \cdot \left[(u_{i_1-1} - u_{i_1+1}) + (u_{i_2-1} - u_{i_2+1}) + \cdots + (u_{i_k-1} - u_{i_k+1}) \right].
    \label{eqn:CycleDeriv}
\end{equation}
Addition in the index terms is done modulo $n$ with adjustments made so that indices run from $1$ to $n$. We now use the contrivance of an imaginary flow from Visomirski and Griffin \cite{visomirski2025integrability}. The sum on the right-hand-side performs in/out flow counting with respect to the non-cycle; that is, we imagine a unit of flow going into and out of each vertex $i_j$. It is now straightforward to see why Itoh-style conserved quantities are conserved. The symmetry in the directed cycle ensures that for non-cycles, all flow counts cancel in derivatives of sums of non-edges of the same size.   

\section{Spoke Graphs}\label{sec:Spoke}


\begin{lemma} The spoked graph $\mathrm{SG}(n,k)$ has interaction matrix with rank $n - 1$ if $n$ is odd and $n - 2$ if $n$ is even.
\end{lemma}
\begin{proof} Note that the rank of the interaction matrix corresponding to the directed cycle with $n$ vertices has rank $n-1$ if $n$ is odd and $n - 2$ if $n$ is even.

Assume that a spoked vertex is connected to vertices $i, i+1, i+2$ and $i+3$. When compared to the interaction matrix of the directed $n$-cycle, the interaction matrix of $\mathrm{SG}(n,k)$ has $k$ additional rows corresponding to the new $k$ additional vertices. Without loss of generality, each spoke vertex generates a column with form,
\begin{equation*}
    \mathbf{c}_{n+j} = \begin{bmatrix}
    0\\
    \vdots\\
    1\\
    1\\
    -1\\
    -1 \\
    \vdots \\
    0\\
    \hline
    \mathbf{0}
    \end{bmatrix},
\end{equation*}
where $j \in \{1,\dots,k\}$. Here $\mathbf{0}$ is a $k$ dimensional zero-vector and we are considering the spoke vertex $n+j$. The horizontal line indicates the separation between directed $n$-cycle vertices and the spoke vertices. We note that the directed cycle structure may cause wrapping in the entries of the column above the horizontal line, hence our assertion we are working without loss of generality. The columns corresponding to vertices $i+1$ and $i+2$, have form,
\begin{equation*}
    \mathbf{c}_{i+1} = \begin{bmatrix}
    0\\
    \vdots\\
    1\\
    0\\
    -1\\
    0 \\
    \vdots \\
    \hline
    -\mathbf{e}_j
    \end{bmatrix} \quad \text{and} \quad
    \mathbf{c}_{i+2} = 
    \begin{bmatrix}
    0\\
    \vdots\\
    0\\
    1\\
    0\\
    -1 \\
    \vdots \\
    \hline
    \mathbf{e}_j
    \end{bmatrix}.
\end{equation*}
This structure is ensured by \cref{def:SpokedCycle} and the fact (as observed) that vertices $i+1$ and $i+2$ may interact with only one spoke vertex. From this, it is clear that $\mathbf{c}_{n+1} = \mathbf{c}_{i+1} + \mathbf{c}_{i+2}$. But this is true for every spoke vertex. Therefore, the rank of the  interaction matrix of $\mathrm{SG}(n,k)$ must be $n - 1$ if $n$ is odd and $n - 2$ if $n$ is even.
\label{lem:SpokeRank}
\end{proof}
Let $n = 2m$ if $n$ is even or $n = 2m+1$ if $n$ is odd. The fact that the directed $n$-cycle is integrable implies that its dynamics admit $m+1$ (commuting) conserved quantities. Similarly, if the dynamics generated by $\mathrm{Sp}(n,k)$ are integrable, then using \cref{lem:SpokeRank} and \cref{def:LAIntegrability} we see that its dynamics must admit $m+k+1$ conserved quantities as a necessary condition for $\mathrm{Sp}(n,k)$ to be integrable.

We can show that each spoke vertex generates a Casimir of the Poisson algebra. As before, assume that a spoke vertex ($n+j$) is connected to vertices $i$, $i+1$, $i+2$ and $i+3$. Let $\mathbf{A}$ be the interaction matrix of $\mathrm{Sp}(n,k)$. Consider the (eigen)vector $\mathbf{v} = \langle {v_1,\dots,v_{n+k}} \rangle$ with,
\begin{equation}
    v_l = \begin{cases}  -1 & \textit{$l = i+1$ \text{or} $l = i+2$}\\
    1 & \text{if $l = n+j$}\\
    0 & \text{otherwise}.
    \end{cases}
    \label{eqn:SpokedEigenvector}
\end{equation}
Direct computation shows that $\mathbf{A}\mathbf{v} = \mathbf{0}$. To see this, note that $\mathbf{v}$ has exactly two non-zero entries that correspond to entries of opposite sign in rows $i$, $i+1$, $i+2$, $i+3$ and $n+j$ in $\mathbf{A}$ and is orthogonal to all other rows. Thus, each spoke vertex produces a Casimir of form,
\begin{equation*}
    C_{n+j} = \frac{u_{n+j}}{u_{i+1}u_{i+2}}.
\end{equation*}
Simultaneously, we see (again from direct computation) that $\tilde{\mathbf{v}} = \langle {\tilde{v}_1,\dots,\tilde{v}_{n+k}} \rangle$ with,
\begin{equation}
    v_l = \begin{cases}  1 & \text{if $i \in \{1,\dots,n\}$}\\
    0 & \text{otherwise},
    \end{cases}
    \label{eqn:MainCasimir}
\end{equation}
is likewise an eigenvector of $\mathbf{A}$. Thus,
\begin{equation*}
    H_{n} = u_1u_2 \cdots u_n,
\end{equation*}
is a Casimir of the Poisson algebra. This allows us to rewrite the Casimirs (conserved quantities) generated by the spoke vertices in Itoh-style,
\begin{equation*}
    \tilde{C}_{n+j} = C_{n+j}H_n = u_1\cdots u_{i}u_{i+3}\cdots u_{n} u_{n+j}.
\end{equation*}
In particular, we can create an equivalent set of Itoh-style conserved quantities of varying polynomial degrees as follows. Let, 
\begin{align*}
    &H_{n-1} = \sum_j H_nC_{n+j}\\
    &H_{n-2} = \sum_{j_1\neq j_2} H_nC_{n+j_1}C_{n+j_2}\\
    &\hspace*{5em}\vdots\\
    &H_{n-k} = H_n \prod_j C_{n+j}.
\end{align*}
Here, $H_l$ has polynomial degree $l$. Trivially, the Hamiltonian, $H_1 = u_1 + u_2 + \cdots u_n + \cdots + u_{n+k}$ is a conserved quantity. Thus, we have $k+2$ conserved quantities. We now show that the remaining $m - 1$ conserved quantities are extracted using non-cycles (and non-edges) from the graph structure, as in \cref{eqn:ItohCC}.

Consider a non-cycle of length $2 \leq l \leq m$ in $B(n,1)$ with polynomial $T = u_{i_1}\cdots u_{i_l}$. We know that, $i_1,\dots,i_l \leq n$ and if none of the vertices $\{i_1,\dots,i_l\}$ are adjacent to a spoked vertex in $\mathrm{Sp}(n,k)$, then,
\begin{equation*}
    \dot{T} = D_{\mathrm{Sp(n, k)}}[T] = D_{B(n,1)}[T].
\end{equation*}
That is, $T$ has the same time derivative in the directed $n$-cycle as in the spoked graph, given in \cref{eqn:CycleDeriv}. Assume now this is not the case for the non-cycle, and consider an arbitrary spoke vertex that is connected to some subset of the vertices in $\{i_1,\dots,i_l\}$. By construction, this spoke vertex is connected to either one or two vertices in $\{i_1,\dots,i_l\}$. Suppose the spoke vertex is connected to two vertices in $\{i_1,\dots,i_l\}$, and without loss of generality let them be $i_1$ and $i_2$ and let the spoke be vertex $n+j$. Then we must have edges of the form $(i_1,n+j)$ and $(n+j,i_2)$ in $\mathrm{Sp}(n,k)$. Or in general, there must be one directed edge from a vertex in $\{i_1,\dots,i_l\}$ to the spoke vertex and one directed edge from the spoke vertex to a vertex in $\{i_1,\dots,i_l\}$. Consequently, by \cref{eqn:CycleDeriv}, $D_{\mathrm{Sp}(n,k)}[T]$ does not contain $u_{n+j}$ (because of cancellation) and so it has the same time derivative as in the directed $n$-cycle as well.

On the other hand, suppose exactly one vertex in $\{i_1,\dots,i_l\}$ is connected to the spoke vertex $n+j$. Then $u_{n+j}$ will necessarily appear in $D_{\mathrm{Sp}(n,k)}[T]$. In particular, if $\bar{C}(k)$ is the set of non-cycles of length $l$ in $B(n,1)$ (the directed $n$-cycle), then each term of,
\begin{equation*}
    D_{\mathrm{Sp}(n, k)}\left[\sum_{c \in \bar{C}(l)} \prod_{j \in c} u_j\right],
\end{equation*}
will contain a spoked vertex variable $u_{n+j}$ (for some $j$) in the product. As a result, we must introduce non-cycles containing the spoked vertices. Let $\bar{C}_{Sp(n,k)}(l)$ be the set of length $l$ non-cycles (including non-edges) in the spoke vertex with $n$ vertices and $k$ spokes. As in the case of the cycle, $2 \leq l \leq \left\lfloor\tfrac{n}{2}\right\rfloor = m$. We now have the conserved quantities,
\begin{equation}
    H_l = \sum_{c \in \bar{C}_{Sp(n,k)}(l)} \prod_{j \in c} u_j,
    \label{eqn:SpokeNonCycle}
\end{equation}
for $2 \leq l \leq m$. We now use the same symmetry argument used in our discussion of the Itoh style conserved quantities of the cycle (with an additional observation) to see that $D_{\mathrm{Sp}(n,k)}\left[H_l\right] = 0$.

As before, symmetry ensures all terms cancel, however there is an extra parity symmetry in play that is not present in the cycle symmetry, and it explains (\textit{post facto}) the definition of spoked graphs. Without loss of generality, assume a spoke occurs at $i_1,i_2,i_3,i_4$. Then $(i_1,i_3)$ is a non-edge as is $(i_2,i_4)$ and these may occur in non-cycles. Notice that $i_1$ and $i_3$ interact with the spoke in opposite ways (i.e., an imaginary flow goes from $i_1$ to the spoke and from the spoke to $i_3$). That is, their edges have opposite parity with respect to the spoke. This symmetry (or balance) coupled with the cycle symmetry ensures cancellation of all terms in $D_{\mathrm{Sp}(n,k)}\left[H_l\right]$. This also explains why spokes have exactly two in edges and two out edges, and why no two distinct spoke vertices can be adjacent to the same vertex in the same way. Modifying the spoking structure from \cref{def:SpokedCycle} would prevent derivatives of polynomials of the form \cref{eqn:SpokeNonCycle} from being zero. Thus, we have shown each spoke graph admits $m+k+1$ conserved quantities.

Using a similar argument structure but applied to the bracket \cref{eqn:Bracket} (or adopting Itoh's proof approach from \cite{I08}), it is clear that the inherent symmetry in the cycle underlying the spoke graph as well as the parity symmetry in the spoke vertices ensures that any two conserved quantities must commute with each other under the action of the bracket \cref{eqn:Bracket}. We state the following lemma, summarizing the discussion above.

\begin{lemma} Let $H_1, H_2,\dots, H_m, H_{n-k},\dots,H_n$ be the polynomials defined above. Then these are the conserved quantities for $\mathrm{Sp}(n,k)$ and furthermore these quantities commute under the action of the nonlinear bracket in \cref{eqn:Bracket}.
\end{lemma}

Lastly, we show that all of these conserved quantities are algebraically independent. Trivially, the conserved quantities inherited from the directed cycle (i.e., non-cycles) are algebraically independent. We now apply the same technique used in Section 3 of Visomirski and Griffin \cite{visomirski2025integrability}. Apply the Jacobian criterion \cite[Theorem 2.3]{P08} to the set of conserved quantities by forming their Jacobian. By construction, each conserved quantity generated from $\mathrm{Sp}(n,k)$ has a distinct polynomial degree. The Hamiltonian has degree $1$. Casimirs have degrees $n$, $n-1, \cdots n-k$, where $k \leq \left\lfloor\tfrac{n}{3}\right\rfloor$. Conserved quantities have degrees $2, 3, \cdots, \left\lfloor\tfrac{n}{2}\right\rfloor$. Consequently, each row of the Jacobian matrix will contain polynomials with unique degree compared to the other rows of the matrix (again by construction of the conserved quantities). Moreover, because each column in the Jacobian will eliminate a variable, it is straightforward to see that Gaussian elimination applied to this Jacobian matrix will yield the identity (as in \cite{visomirski2025integrability}). Thus, by the Jacobian criterion, we must have an algebraically independent set.  

\begin{theorem} Let $G = \mathrm{Sp}(n,k)$ with $k \leq \left \lfloor \tfrac{n}{3}\right\rfloor$. Then $G$ is integrable.
\end{theorem}

We illustrate the process described above with the simple spoke graph generated from the directed 5-cycle with a single spoke (see \cref{fig:SimpleSkipGraph}). This graph was treated as an outlier by Visomirski and Griffin in \cite{visomirski2025integrability}.
\begin{figure}[htbp]
\centering
\includegraphics[width=0.45\textwidth]{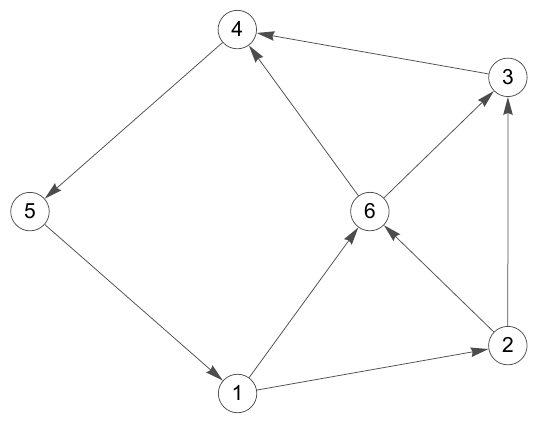}
\caption{Directed 5 cycle with a single spoke ($\mathrm{Sp}(5,1)$)}
\label{fig:SimpleSkipGraph}
\end{figure}
The Hamiltonian is,
\begin{equation*}
    H_1 = u_1 + u_2 + u_3 + u_4 + u_5 + u_6.
\end{equation*}
There are two Casimirs,
\begin{equation*}
    H_5 = u_1 u_2 u_3 u_4 u_5,
\end{equation*}
from the cycle and,
\begin{equation*}
    C = \frac{u_6}{u_2u_3},
\end{equation*}
which generates,
\begin{equation*}
    H_4 = C \cdot H_5 = u_1u_4u_5u_6.
\end{equation*}
The remaining conserved quantity is generated from the non-edges and is,
\begin{equation*}
    H_2 = u_1u_3 + u_1u_4 + u_2u_4 + u_2u_5 + u_3u_5 + u_5u_6.
\end{equation*}
Notice, all the terms in $H_2$ except $u_5u_6$ would appear in the corresponding conserved quantity in the directed 5-cycle. The $u_5u_6$ term corrects this conserved quantity to account for the addition of the spoke vertex. Code to compute the conserved quantities for arbitrary spoke graphs is provided in the supplemental information (SI).

Interestingly, the eigenvectors leading to the Casimirs in the previous discussion also show the existence of an infinite set of interior fixed points in the dynamics generated by a spoked graph. Assuming there are $k$ spoked vertices and $n$ vertices in the underlying cycle, let $\mathbf{v}_{n+1},\dots,\mathbf{n+k}$ denote the eigenvectors with form given in \cref{eqn:SpokedEigenvector} and as before let $\tilde{v}= \langle{1,\dots,1,0,\dots,0}\rangle$ be the eigenvector corresponding to the Casimir $H_n = u_1\cdots u_n$ and given in \cref{eqn:MainCasimir}. Then,
\begin{equation*}
    \mathbf{u}^* = r\left(\mathbf{v}_{n+1} + \cdots + \mathbf{v}_{n+k}\right) + a\tilde{\mathbf{v}},
\end{equation*}
necessarily has the property that $\mathbf{A}\mathbf{u}^* = \mathbf{0}$. Consequently, if we choose $a$ and $r$ so that $\mathbf{u}^* \in \mathrm{int}\,\Delta_{n+k-1}$, then we have found an interior fixed point. Trivially, we require $0 < a < 1$ and $0 < r < 1$, since some indices of $\mathbf{u}$ will be precisely equal to those values by construction, while the remaining indices will have value $a - r$ and therefore we must have $r < a$. Finally, to ensure the proposed solution is on the simplex we have,
\begin{equation*}
    \sum_i u_i^* = na - rk = 1.
\end{equation*}
A straightforward calculation shows that if we have,
\begin{align*}
    & 0 < r < \frac{1}{n-k} \quad \text{and}\\
    &a = \frac{1+kr}{n},
\end{align*}
then the resulting fixed point will be in the interior of the relevant unit simplex. Thus, unlike the usual case where an integrable system admits a single (elliptic) fixed point, here we have an infinite number of fixed points with their own local dynamics, necessarily leading to integrable dynamics.

\section{Anti-Cloning}\label{sec:AntiCloning}

Previous work by \cite{EKV21} \cite{EKV22} defined the \textit{cloning} operation, in which  an additional vertex is added to a graph having the same neighbourhood as an existing vertex, and showed that the operation preserves graph integrability. Visomirski and Griffin \cite{visomirski2025integrability} showed that cloning is a special case of a kind of graph embedding, generalising these results. In contrast, we will show that the anti-cloning operation preserves integrability only when iteratively applied to a balanced tournament. The approach taken in this section mirrors the approach taken in the previous section.

Before proceeding, recall that the balanced five vertex tournament, with a single anti-cloned vertex is illustrated in \cref{fig:AntiCloning}. It is worth noting that this is not the smallest example of an anti-cloned tournaments. Anti-cloning a vertex of the balanced 3-tournament (or the directed 3-cycle) is the most basic example of this family (see \cref{fig:seedAC}) . 
\begin{figure}
    \centering
    \includegraphics[width=0.5\linewidth]{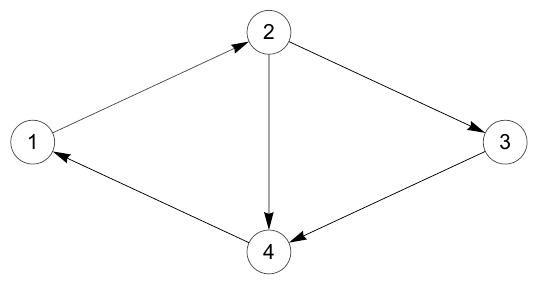}
    \caption{Most basic example of the Anti-Cloned Tournament graph family. Note that it is also the most basic example of the Skip Vertex graph family}
    \label{fig:seedAC}
\end{figure}
However, this graph also represents the most basic example of the skip-vertex graphs previously introduced by Visomirski and Griffin \cite{visomirski2025integrability}. This overlap prevented immediate identification of this new family of graphs in previous work. 

\begin{lemma} Let $G=\mathrm{To}(n)$ be a balanced tournament on $n = 2m+1$ vertices with vertex set $\{1,\dots,n\}$. Suppose $H = \mathrm{Ac}(G,{i_1,\dots,i_k})$ and denote the (distinct) anti-cloned vertices of $i_1,\dots,i_k$ by $n+1,\dots,n+k$. Then, the dynamics generated by $H$ admit the Casimirs,
\begin{equation*}
    C_{i_ji'_j} = u_{i_j}u_{n+j},
\end{equation*}
for each pair $(i_j,n+j)$ with $j \in \{1,\dots,k\}$.
\label{lem:AntiCloneCasimir}
\end{lemma}
\begin{proof} Let $\mathbf{A}$ be the interaction matrix generated from $H$. The lemma follows immediately from the fact that column $i_j$ and $n+j$ in $\mathbf{A}$ are additive inverses and therefore generate the eigenvector $\mathbf{v}_j$ with components,
\begin{equation*}
    v_{j_l} = \begin{cases} 1 & \text{if $l = j$ or $l = j'$}\\
    0 & \text{otherwise}.
    \end{cases}
\end{equation*}
Thus, $u_{i_j}u_{n+j}$ is a Casimir for all $j \in \{1,\dots,k\}$.
\end{proof}

It is worth noting that this lemma holds for anti-clones in arbitrary graphs, and not just balanced tournament graphs. However, as we will see, the balanced tournaments have the property that anti-cloning preserves the structure of the existing conserved quantities. 

\begin{proposition} Let $G=\mathrm{To}(n)$ be a balanced tournament on $n = 2m+1$ vertices with vertex set $\{1,\dots,n\}$. Suppose (without loss of generality) $H = \mathrm{Ac}(G,\{1,\dots,k\})$ and denote the anti-cloned vertices of $1,\dots,k$ by $n+1,\dots,n+k$. Then the dynamics generated by $H$ admit $m+k+1$ conserved quantities, a sufficient number for integrability. 
\label{prop:AntiClone}
\end{proposition}
\begin{proof} From \cref{lem:AntiCloneCasimir} we have identified $k$ Casimirs. By construction, we have the Hamiltonian $H_1 = u_1 + \cdots u_n$, which is a conserved quantity. The interaction matrix of $H$ will have $k$ additional rows (and columns) not present in the interaction matrix of $G$. From \cite{PG23,I87}, we know that $H_n = u_1\cdots u_n$ is both a conserved quantity and Casimir because,
\begin{equation*}
    \mathbf{v}_n = \langle{\underbrace{1,1,\dots,1}_n}\rangle
\end{equation*} is an eigenvector of the interaction matrix of $G$. It follows at once that,
\begin{equation*}
    \mathbf{v}_n' = \langle{\underbrace{1,1,\dots,1}_n,\underbrace{0,\cdots,0}_k}\rangle
\end{equation*}
is an eigenvector of the interaction matrix corresponding to $H$ because we have simply zeroed out the columns corresponding to the added vertices. Thus, we have identified $k+2$ conserved quantities. 

The remaining $m-1$ conserved quantities are composed of the sums of Casimirs of the sub-tournaments of $H$, just as in the case of the balanced tournaments. To see this, suppose $\mathcal{T}_H(2l+1)$ represents the set of vertex sets of subtournaments of $H$ of size $2l+1$ and let $\mathcal{T}_G(l)$ be similarly defined for $G$. If we have,
\begin{equation*}
    C_{2l+1}= \sum_{t \in \mathcal{T}_G(2l+1)}\prod_{j\in t}u_j,
\end{equation*}
then because $C_{2l+1}$ is conserved in the dynamics generated by $G$ it is clear that $D_H[C_{2l+1}]$ must have terms remaining that contain $u_j$ for $j \in \{n+1,\dots,n+k\}$. That is, the derivative has anti-cloned vertex terms. In fact, those terms will (in a sense) characterize the imaginary flow introduced by the addition of the anti-clone vertices. If we replace $C_{2l+1}$ with,
\begin{equation*}
    H_{2l+1}= \sum_{t \in \mathcal{T}_H(2l+1)}\prod_{j\in t}u_j,
\end{equation*}
which contains additional terms corresponding to subtournaments formed with the anti-cloned vertices, then the anti-symmetry of the edges leading to and from the anti-cloned vertices will cancel the unaccounted for terms in the derivative of $C_{2l+1}$. Thus, $H_{2l+1}$ is a conserved quantity for $H$ both by the symmetry inherited from $G$ in the edge structure on vertices $\{1,\dots,n\}$ and the anti-symmetry of the edge structure on vertices $\{n+1,\dots, n+k\}$. 

We know that there are $m-1$ proper subgraphs isomorphic to a balanced tournament in $G$. Consequently, $H$ must have $m-1$ proper subgraphs isomorphic to a balanced tournament. Thus, we have identified $m-1$ additional conserved quantities corresponding to the subtournaments of $H$ for a total of $m+k+1$ total conserved quantities. The fact that the rank of the matrix corresponding to $G$ is $2m$ and adding anti-cloned vertices does not increase this rank by \cref{lem:AntiCloneCasimir} implies that we require $m+k+1$ conserved quantities for the integrability of $H$.
\end{proof}

As in our discussion on spoke graphs, we note that we can modify the conserved quantities of $\mathrm{Ac}(2m+1,k)$ so that they each have distinct degree. Assume we have $k$ anti-cloned vertices. Replace the simple Casimirs of form $C_{n+j} = u_{i_j}u_{n+j}$ identified in the proof of \cref{lem:AntiCloneCasimir} with Casimirs of the form,
\begin{align*}
    &H_{2} = \sum_j C_{n+j}\\
    &H_{4} = \sum_{j_1\neq j_2} C_{n+j_1}C_{n+j_2}\\
    &\hspace*{5em}\vdots\\
    &H_{2k} = \prod_j C_{n+j}.
\end{align*}
These polynomials all have even degree by definition. At the same time, the remaining conserved quantities $H_1$, $H_n$ and $H_{2l+1}$ identified in the proof of \cref{prop:AntiClone} all have odd degree. Thus, these conserved quantities will be independent by the Jacobian criterion. Finally, that these all commute under the action of the bracket can be seen by symmetry argument (as with the spoke graphs or as in the above discussion) or via an Itoh style argument \cite{I08}. Thus we may summarize the above discussion in the following theorem.

\begin{theorem} Let $G = \mathrm{Ac}(2m+1, k)$, where $k \in \{0,\dots,2m+1\}$. Then $G$ is integrable. 
\end{theorem}

By way of example, consider again the anti-cloned graph shown in \cref{fig:AntiCloning}. This graph appeared as an outlier in our earlier work \cite{visomirski2025integrability}. It is now (reasonably) easy to read off the conserved quantities of the graph. We have the Hamiltonian,
\begin{equation*}
    H_1 = u_1 + u_2 + u_3 + u_4 + u_5 + u_6,
\end{equation*}
and the major Casimir,
\begin{equation*}
    H_{5} = u_1u_2u_3u_4u_5.
\end{equation*}
The simple Casimir is,
\begin{equation*}
    H_2 = u_1u_6,
\end{equation*}
because there is only one anti-cloned vertex. Finally, the remaining conserved quantity is composed of the Casimirs of the subtournaments containing three vertices,
\begin{equation*}
    H_3 = \underbrace{u_1 u_2 u_3+u_2 u_3 u_4 + u_3 u_4 u_5 + u_1 u_4 u_5 + u_1 u_2 u_5}_{C_3} + \underbrace{u_3 u_4 u_6}_\text{Correction}.
\end{equation*}
Code to compute these conserved quantities for arbitrary anti-clone graphs is provied in the SI.

We conclude by noting the critical difference between the cloning operation of \Evri et al. and this anti-cloning operation. As \Evri et al. showed in \cite{EKV22}, cloning a vertex of \textit{any} integrable graph will produce another integrable graph. However, the anti-cloning operation does not generalise in the same way. Anti-cloning the vertex of any non-tournament integrable graph (i.e. all other \Bogo graphs, skip vertex graphs, spoked graphs, etc.) does not produce an integrable graph (and instead shows evidence that it produces a graph whose dynamics are chaotic). In fact, this operation seems to be the first one that consistently transforms integrable graphs (except balanced tournaments) into graphs with chaotic dynamics. By way of  example, when we anti-clone vertex 1 of the 5-cycle, see \cref{fig:ACChaos} left, we see immediate numerical evidence for chaos as shown by Lyapunov exponent computation in \cref{fig:ACChaos} (right).
\begin{figure}[htbp]
\centering
\includegraphics[width=0.45\textwidth]{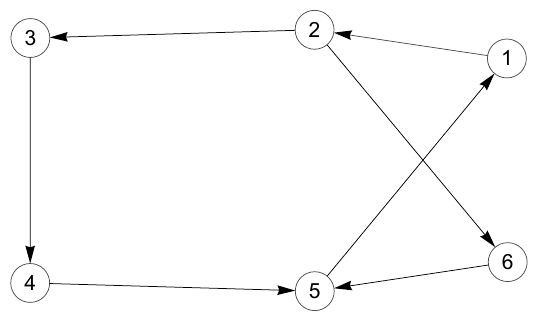}\quad
\includegraphics[width=0.45\textwidth]{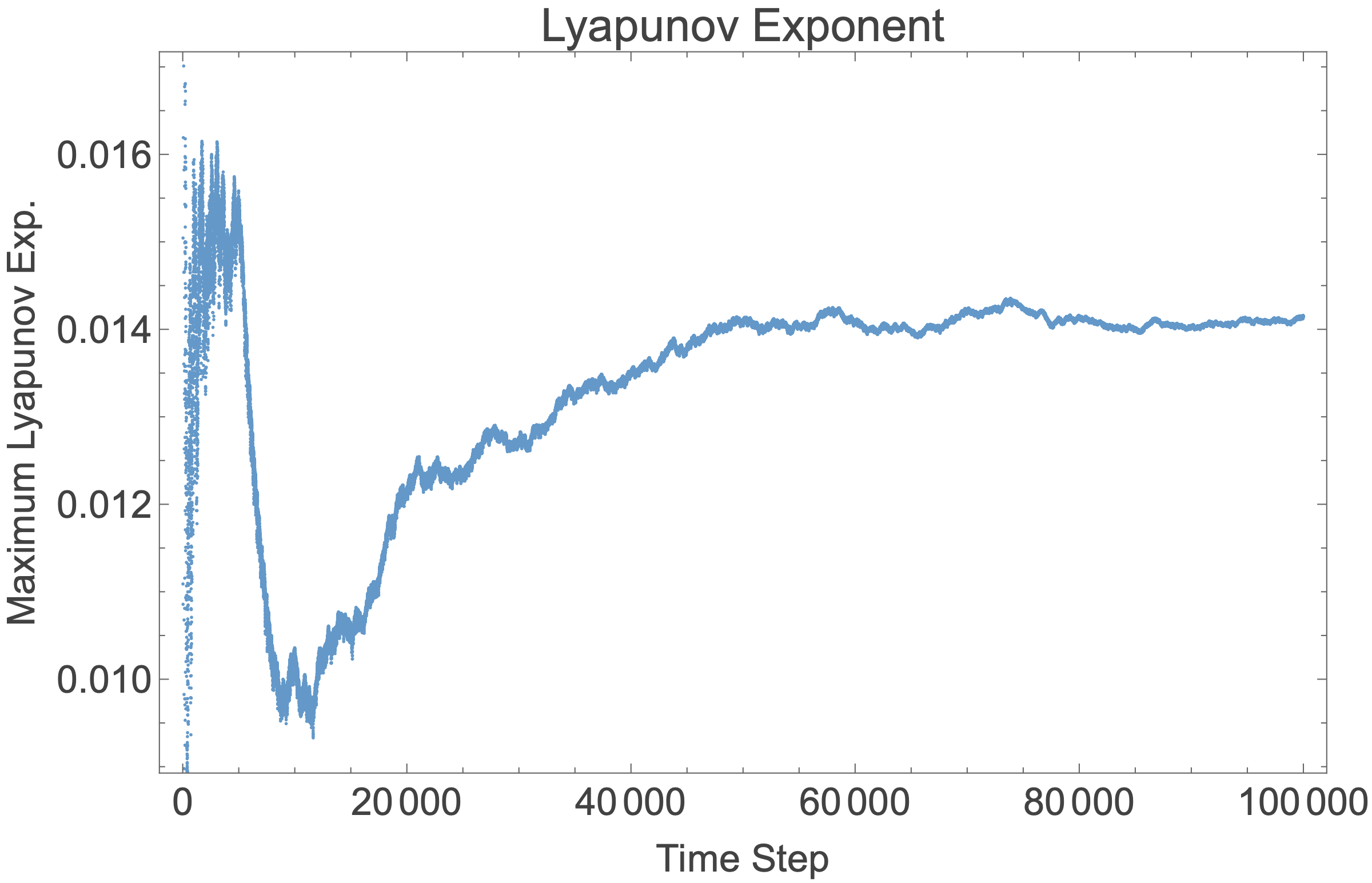}
\caption{Example of an anti-cloned 5-cycle with numerical evidence of chaos}
\label{fig:ACChaos}
\end{figure}

\section{Conjectures on \Bogo Graphs}\label{sec:Bogo}


We briefly turn our attention to the \Bogo graphs. As noted, specific subfamilies are known to be integrable (e.g., $B(n, 1)$ and $\mathrm{To}(2m+1)$) and it has been conjectured that the entire family produces integrable dynamics \cite{visomirski2025integrability}. In this section, we seek to relate our observations in this paper to the \Bogo graphs with the goal of providing a framework for a proof of the integrability of the entire family. We do not attempt the proof in this paper. 

We begin by observing that every \Bogo graph with $n$ vertices has conserved quantities,
\begin{equation*}
    H_1 = \sum_{i=1}^n u_i \qquad\text{and} \qquad
    H_n = \prod_{i=1}^n u_i,
\end{equation*}
and we look for additional conserved quantities expressible using graph-theoretic (i.e., combinatorial) language. We make the following conjecture relating the conserved quantities of the \Bogo graphs to their graph structure. 
\begin{conjecture}
A complete set of conserved quantities of any \Bogo graph $G$ can be constructed from a combination of two sets: $\mathcal{C}_0$ and $\mathcal{C}_1$. If $\mathcal{C}_0$ is non-empty, then it is a set of polynomials whose terms are constructed from index sets of vertices contained in non-edges and non-cycles of $G$ as in \cref{eqn:ItohCC} or \cref{eqn:ItohCC2} and if $\mathcal{C}_1$ is non-empty, then it contains  sums of the Casimirs of edge maximal sub-\Bogo graphs contained $G$.
\label{conj:bogoccsconjecture}
\end{conjecture}

As noted, the conserved quantities of $B(n, 1)$ (the directed cycles) consist entirely of polynomials generated by vertex index sets consisting corresponding to non-edges and non-cycles of the underlying graph. In contrast, the conserved quantities of $\mathrm{To}(2m+1)$ (the complete balanced tournaments) consist of polynomials whose terms correspond to subtournaments of the graph in question. 

By way of example, consider the \Bogo graph in \cref{fig:B72}. In addition to $H_0$ and $H_7$, the other conserved quantities are,
\begin{align*}
    &H_2 = u_1u_4 + u_1u_5 + u_2u_5 + u_2u_6 + u_3u_6 + u_3u_7 + u_4u_7\\
    &H_4 = u_1u_2u_4u_6 + u_2u_3u_5u_7 + u_3u_4u_6u_1 + u_4u_5u_7u_2 + u_5u_6u_1u_3 +\\
    &\hspace*{24em}u_6u_7u_2u_4 + u_7u_1u_3u_5.
\end{align*}
Here, the terms of $H_2$ are indexed by the vertices in the non-edges of $B(7, 2)$. On the other hand, the terms in $H_4$ are indexed by vertices in the directed 4-cycle (i.e., $B(4, 1)$) subgraphs in $B(7,2)$. These are the only sub-\Bogo graphs of $B(7,2)$ and hence the maximal such subgraphs.

Continuing our argument by example, consider $B(9, 2)$, shown in \cref{fig:B92}. 
\begin{figure}[htbp]
\centering
\includegraphics[width=0.35\textwidth]{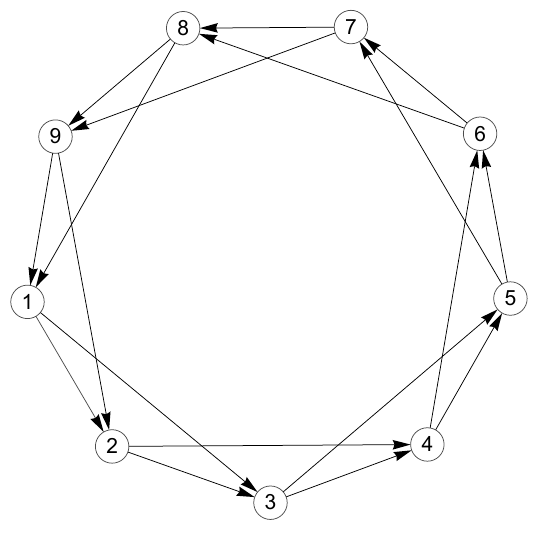}
\caption{The \Bogo graph $B(9,2)$}
\label{fig:B92}
\end{figure}
In addition to $H_1$ and $H_9$, it's conserved quantities are:
\begin{align*}
    &H_2 = u_1u_4 + u_1u_5 + u_1u_6 + u_1u_7 + u_2u_5 + u_2u_6 + u_2u_7 + u_2u_8 + u_3u_6 + u_3u_7 + u_3u_8 +\\ 
    &\hspace*{17em}u_3u_9 + u_4u_7 + u_4u_8 + u_4u_9 + u_5u_8 + u_5u_9 + u_6u_9\\
    &H_3 = u_1u_4u_7 + u_2u_5u_8 + u_3u_6u_9\\
    &H_5 = u_1u_2u_4u_6u_8 + u_1u_3u_4u_6u_8 + u_1u_3u_5u_6u_8 + u_1u_3u_5u_7u_8 + u_1u_3u_5u_7u_9 +\\
    &\hspace*{15em}u_2u_3u_5u_7u_9 + u_2u_4u_5u_7u_9 + u_2u_4u_6u_7u_9 + u_2u_4u_6u_8u_9\\
    &H_6 = u_1u_2u_4u_5u_7u_8 + u_1u_3u_4u_6u_7u_9 + u_2u_3u_5u_6u_8u_9
\end{align*}

This is the smallest graph that admits multiple conserved quantities in both $\textbf{C}_0$ and $\textbf{C}_1$. Its conserved quantities are the non-edges and non 3-cycles, as well as the sum of the Casimirs of the sub-$B(5, 1)$ graphs and sub-$B(6, 1)$ graphs in $B(9,2)$. These are edge maximal sub-\Bogo graphs. However, notice that there are multiple isomorphic copies of these graphs in $B(9,2)$ that are not all represented in the polynomials $H_5$ and $H_6$. Determining the criteria by which the isomorphic subgraphs are used in a question for future research and would be needed to complete \cref{conj:bogoccsconjecture}.

Finally, consider $B(10, 4)$ in \cref{fig:B104}. 
\begin{figure}[htbp]
\centering
\includegraphics[width=0.35\textwidth]{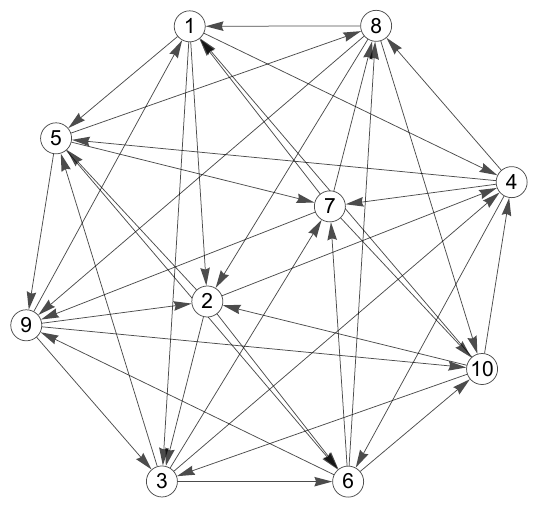}
\caption{The \Bogo graph $B(10,4)$. In addition to having cycle subgraphs (3-cycle, 4-cycle) and tournament subgraphs ($B(5, 2)$), it also has sub-\Bogo graphs ($B(6, 2), B(8, 3)$)}
\label{fig:B104}
\end{figure}
We consider this graph because some of its conserved quantities are derived from sub-\Bogo graphs that are neither cycles nor complete tournaments. Aside from $H_1$ and $H_{10}$, it's conserved quantities are,
\begin{align*}
    &H_2 = u_1u_6 + u_2u_7 + u_3u_8 + u_4u_9 + u_5u_{10}\\
    &H_3 = u_1u_3u_7 + u_1u_4u_7 + u_1u_5u_7 + u_1u_4u_8 + u_2u_4u_8 +u_1u_5u_8 + u_2u_5u_8 + u_2u_6u_8 +\\
    &\hspace*{8em}u_1u_5u_9 + u_2u_5u_9 + u_3u_5u_9 + u_2u_6u_9+u_3u_6u_9 + u_3u_7u_9 + \\
    &\hspace*{12em}u_2u_6u_{10} + u_3u_6u_{10} + u_4u_6u_{10} + u_3u_7u_{10} + u_4u_7u_{10} + u_4u_8u_{10}\\
    &H_4 = u_1u_2u_6u_7 + u_1u_3u_6u_8 + u_2u_3u_7u_8 + u_1u_4u_6u_9 + u_2u_4u_7u_9 + u_3u_4u_8u_9 +  \\ 
    &\hspace*{17em}u_1u_5u_6u_{10} + u_2u_5u_7u_{10} + u_3u_5u_8u_{10} + u_4u_5u_9u_{10}\\
    &H_5 = u_1u_3u_5u_7u_9 + u_2u_4u_6u_8u_{10}\\
    &H_6 = u_1u_2u_3u_6u_7u_8 + u_1u_2u_4u_6u_7u_9 + u_1u_3u_4u_6u_8u_9 + u_2u_3u_4u_7u_8u_9 + \\
    &\hspace*{4em} u_1u_2u_5u_6u_7u_{10} +
    u_1u_3u_5u_6u_8u_{10} + u_2u_3u_5u_7u_8u_{10} + u_1u_4u_5u_6u_9u_{10} + \\
    &\hspace*{4em}u_2u_4u_5u_7u_9u_{10} + u_3u_4u_5u_8u_9u_{10}
    u_1u_2u_3u_4u_6u_7u_8u_9 + u_1u_2u_3u_5u_6u_7u_8u_{10} + \\
    &\hspace*{20em}u_1u_2u_4u_5u_6u_7u_9u_{10} + u_1u_3u_4u_5u_6u_8u_9u_{10}
\end{align*}
These correspond non-cycles or sub-graphs as,
\begin{enumerate}
    \item The conserved quantity $H_2$ is generated from vertex index sets corresponding to the non-edges of $B(10,4)$.
    \item The conserved quantity $H_3$ is generated from vertex index sets corresponding to the three-cycles of $B(10,4)$.
    \item The conserved quantity $H_4$ is generated from vertex index sets corresponding to a subset of four-cycles $B(10,4)$. The directed 4-cycle is the maximal four vertex \Bogo subgraph of $B(10,4)$.
    \item The conserved quantity $H_5$ is generated from vertex index sets corresponding to a subset of the balanced 5-tournament subgraphs of $B(10,4)$.
    \item The conserved quantity $H_6$ is generated from vertex index sets corresponding to a subset of the $B(6,2)$ subgraphs of $B(10,4)$. These are the edge maximal six vertex \Bogo subgraphs that occur in $B(10,4)$.
\end{enumerate}

For $B(7, 2)$, $B(9, 2)$, and $B(10, 4)$ we have found a sufficient number of conserved quantities to ensure integrability (as expected from  \cref{theorem:bogoccs}). Additionally, these conserved quantities are evidently algebraically independent because they are derived from different graph structures and are all of different order. Finally, direct computation shows that all the conserved quantities for these graphs commute under the action of the bracket. We can therefore conclude that $B(7, 2)$, $B(9, 2)$, and $B(10, 4)$ produce integrable dynamics and conform to the pattern given in \cref{conj:bogoccsconjecture}.

While we have shown that these three graphs produce integrable dynamics, we are unable to generalise this proof for all $B(n, k)$. Even though \cref{conj:bogoccsconjecture} provides the general form of the conserved quantities, it fails to provide an exact prescription for which maximal \Bogo subgraphs to use for a given graph. In other words, it fails to identify what the order ($n$) and edge structure ($k$) of the expected sub-\Bogo graphs will be. Consequently, we do not know what the conserved quantities for an arbitrary \Bogo graph would be when written in combinatorial terms. For small \Bogo graphs, it is easy to simply proceed by exhaustion, and code is available for this in the SI. For arbitrarily large \Bogo graphs, it would be exceedingly difficult to find every family of sub-\Bogo graphs (and therefore each conserved quantity). Despite these limitations, we make the following conjecture.

\begin{conjecture} Assume \cref{conj:bogoccsconjecture} holds. Then the set of conserved quantities as given in the conjecture for any \Bogo graph are algebraically independent and commute under the action of the bracket (\cref{eqn:Bracket}). Consequently, the family of \Bogo graphs is integrable.
\label{conj:bogocommute}
\end{conjecture}

Proving the two conjectures would prove that the \Bogo family is an infinitely large family of graphs that produce integrable dynamics, consistent with all observations to date. We suspect that a proof of \cref{conj:bogoccsconjecture} and \cref{conj:bogocommute} could make use of the (imaginary) flow balance that has been used in this paper and in \cite{visomirski2025integrability}. In fact, it may be possible that identifying the subgraph needed to build conserved quantities can be accomplished by analysing this flow balance property. This is (of course) left to future work. 

\section{Complete Characterization of All Six Vertex Graphs}\label{sec:Taxonomy}

Visomirski and Griffin \cite{visomirski2025integrability} provide a taxonomy of the known integrable graphs with up to six vertices (see Section 6, Table 2 of \cite{visomirski2025integrability}). We now update this taxonomy to include the new families of graphs introduced within this paper. 
\begin{table}[htp!]
    \centering
    \begin{tabular}{lccccc}
    &\multicolumn{4}{c}{Vertex size} & \\ 
    \hline
        Family & 3 & 4 & 5 & 6 & Total\\
        \hline
        $B(n, k)$ & 1 & 1 & 2 & 2 & 6\\
        Skip-Vertex & 0 & 1 & 0 & 3 & 4\\
        Cloned & 0 & 1 & 7 & 24 & 32\\
        Embedding & 0 & 0 & 1 & 11 & 12\\
        Spoke & 0 & 0 & 1 & 1 & 2 \\
        Anti-Cloned & 0 & 0 & 0 & 1 & 1 
    \end{tabular}
    \caption{Updated Breakdown of Integrability by Family and Vertex Size}
    \label{tab:graphfamilies}
\end{table}
We note that several graphs belong to multiple families. These overlaps are not reflected within this table. Instead, graphs are assigned to the first family by order of discovery (e.g., cloned vertex graphs \cite{EKV22} precede, embedded graphs, which precede anti-cloned graphs etc.).

\section{Conclusions and Future Directions}\label{sec:Conclusions}

In this paper, we introduced two new families of integrable graphs, extending the taxonomy introduced by Visomirski and Griffin \cite{visomirski2025integrability} as well as the larger body of work on integrable Lotka-Volterra systems \cite{KM75, M74,I08,BIY08,EKV21,EKV22,PG23}. We defined a family of graph structures, the \textit{spoked graphs}, identified the conserved quantities for the dynamical systems generated by each family member and consequently showed that an arbitrary spoked graph is integrable (i.e., produces integrable dynamics). Next, we introduced a new graph operation, vertex anti-cloning, that is similar to the cloning operation originally defined by \Evri et al. in \cite{EKV22}. We showed that this  vertex operation, generates a new family of integrable graphs when it is applied to balanced tournaments. However, we also found that this operation did not generalise in the same way that cloning did and presented numerical evidence to suggest that applying the anti-cloning operation to a non-tournament integrable graph will create a graph that produces chaotic dynamics. Finally, we provided two conjectures whose proof would establish the integrability of \Bogo graphs using the graph structure-based approach. 

There are several future directions for this research. First, notice in our analysis we have not used Lax pairs to identify conserved quantities or prove integrability as would be typical. One area of future research could be dedicated to determining these Lax pairs for all the integrable families of graphs discussed here and in \cite{visomirski2025integrability}. Furthermore,  proving the conjectures raised in \cref{sec:Bogo} would result in a proof of the full integrability of the \Bogo family of graphs. Additionally, it would be worthwhile to further investigate the anti-cloning operation, specifically in the non-tournament case. The anti-cloning operation seems to transform most integrable graphs to chaotic graphs. Formally proving that the resulting maximal Lyapunov from such a graph would both show that chaos can emerge in this class of dynamics and that anti-cloning is an operation that breaks integrability. To the authors' knowledge, an operation transforming integrable graphs to chaotic graphs would be novel, and a general result on anti-cloning of (e.g.) directed cycles would represent the first known family of chaotic graphs.  Finally, while we have fully completed the taxonomy for all graphs with up to six vertices, we have no way to tell whether there are more integrable families which may start with seven or more vertices. The new families we have found, combined with the graph derivative operator and its properties, could perhaps provide a clue for a potential unified theory of integrability for these graphs.

\bibliographystyle{iopart-num}
\bibliography{Integrability2}
\end{document}